\documentclass[11pt,a4paper]{article}

\usepackage[utf8]{inputenc}
\usepackage[english]{babel}

\usepackage{amsmath}
\usepackage{amsfonts}
\usepackage{amssymb}
\usepackage{amsthm}
\usepackage[T1]{fontenc}
\usepackage{color}
\usepackage{comment}
\usepackage{mathtools}
\usepackage{braket}
\usepackage[explicit]{titlesec}
\usepackage{soul}
\usepackage[toc,page]{appendix}
\usepackage{amsbsy}
\usepackage{upgreek}
\usepackage[all]{xy}
\usepackage{blkarray}

\usepackage{enumitem}
\setlist{nosep}

\usepackage[pdftex,colorlinks=true,
pdfstartview=FitV,
linkcolor= blue,
citecolor= blue,
urlcolor= blue,
hyperindex=true,
hyperfigures=false]
{hyperref}

\usepackage{fancyhdr}

\let\origappendix\appendix 
\renewcommand\appendix{ \clearpage\pagenumbering{roman}\origappendix}

\usepackage[libertine]{newtxmath}

\DeclareFontFamily{OT1}{pzc}{}
\DeclareFontShape{OT1}{pzc}{m}{it}{<-> s * [1.15] pzcmi7t}{}
\DeclareMathAlphabet{\mathpzc}{OT1}{pzc}{m}{it}

\newcommand{\nid}{\noindent}

\newcommand{\tbf}{\mathbf{t}}

\newcommand{\qbf}{\mathbf{q}}
\newcommand{\sbf}{\mathbf{s}}

\newcommand{\Drm}{\mathrm{D}}

\renewcommand{\Wr}{\mathrm{Wr}}

\newcommand{\e}{\mathrm{e}}

\newcommand{\lp}{\left(}
\newcommand{\rp}{\right)}
\newcommand{\lb}{\left[}
\newcommand{\rb}{\right]}

\newcounter{thmcounter}[section]
\numberwithin{thmcounter}{section}

\theoremstyle{plain}
\newtheorem{thm}[thmcounter]{Theorem}

\newtheorem{lem}[thmcounter]{Lemma}
\newtheorem{prop}[thmcounter]{Proposition}

\theoremstyle{definition}

\theoremstyle{remark}
\newtheorem{rmq}[thmcounter]{Remark}

\begin{document}

\thispagestyle{empty}

\vspace{1.5cm}

\begin{center}

~\vspace{1cm}

\begin{LARGE}
\textbf{From the Adler--Moser polynomials\\to the polynomial tau functions of KdV}\\~\\
\end{LARGE}

\begin{large}
Ann du Crest de Villeneuve\\~\\
LAREMA, UMR CNRS 6093, Université d'Angers,\\
2 boulevard Antoine de Lavoisier, 49000 Angers, France,\\
e-mail: ducrest@math.univ-angers.fr
\end{large}

\vspace{1.5cm}

\begin{minipage}{13cm}

\nid \textbf{Key words:} \textit{KdV hierarchy, KP hierarchy, Sato's equation, Adler--Moser polynomials, Tau functions, Rational solutions, Wronskian solutions.}\\

\nid In 1978, M.\! Adler and J.\! Moser proved that there exists a unique change of variables that transforms the Adler--Moser polynomials into the polynomial tau functions of the KdV hierarchy. In this paper we exhibit this change of variables.

\end{minipage}
\end{center}

\setcounter{page}{1}

\vspace{0.5cm}

	\section*{Introduction}

The Korteweg--de Vries hierarchy (or KdV) is a sequence of pairwise commuting partial differential equations first discovered by A.\! Lenard in 1967. Each equation admits a Lax pair as follows \cite{Lax68}. Let a function $u$ which depends on infinitely many variables indexed by odd integers $\tbf = (x=t_1,t_3,t_5,\ldots)$. We denote $\Drm\, u = \partial_x u = u'$. We associate to $u$ the Schrödinger operator $L = \Drm^2 + u$. Then the KdV hierarchy is defined by the flows
\begin{align*}
\frac{\partial L}{\partial t_{2i-1}} = \left[ \left( P^{2i-1} \right)_+, L\right],
\end{align*}
where $P = L^{1/2}$ is the square root of $L$ and for any pseudo-differential operator $X$, $X_+$ is its purely differential part. The first equation reads $\partial_{t_1} u = u'$ (and we set $x=t_1$), the next few equations read
\begin{align*}
\partial_{t_3} u &= \textstyle \frac{1}{4}u^{(3)} + \frac{3}{4}uu',\\
\partial_{t_5} u &= \textstyle \frac{1}{16}u^{(5)} + \frac{4}{8}uu^{(3)} + \frac{5}{4}u'u'' + \frac{15}{8}u^2u',\\
\partial_{t_7} u &= \textstyle \frac{1}{64}u^{(7)} + \frac{7}{32}uu^{(5)} + \frac{21}{32} u'u^{(4)} + \frac{35}{32}u''u^{(3)} + \frac{35}{32}u^2u^{(3)} + \frac{35}{8}uu'u'' + \frac{35}{32}(u')^3 + \frac{35}{16}u^3u',
\end{align*}
the flow with respect to $t_3$ being the KdV equation. To any solution $u$ we associate a tau function defined, up to a multiplicative constant, by
$$u = -2(\log\tau)''.$$

If $\tau$ is polynomial then $u$ is rational. On the other hand, it was proven by Airault, McKean and Moser \cite{AMcKM77} that there are denumerably many rational solutions of KdV, each one being the orbit of a single function $u_n = \frac{n(n+1)}{x^2}$ under the flows of the hierarchy. In \cite{AM78}, the authors constructed the Adler--Moser polynomials $\theta_n(x=q_1,q_3\ldots, q_{2n-1})$ for $n\geq 0$, defined by the recursion
\begin{align*}
\theta'_{n+1}\theta_{n-1} - \theta_{n+1}\theta_{n-1}' = (2n-1)\theta_n^2.
\end{align*}
An important result of \cite{AM78} (cf. Theorem \ref{thmAM}) is that there exists a unique change of variables that transforms the Adler--Moser polynomials into polynomial tau functions of KdV and that we recover all rational solutions of KdV. But we did not know what this change of variables was.

In this paper, we show that the following change of variables transforms the Adler--Moser polynomials into the polynomial tau functions of KdV: $q_1=t_1=x$, and
\begin{align*}
\sum_{i\geq 2} \frac{q_{2i-1}}{\alpha_{2i-1}} z^{2i-1} = \tanh \lp \sum_{i\geq 2} t_{2i-1} z^{2i-1} \rp.
\end{align*}
where $\alpha_{2i-1} = (-1)^{i-1} 3^25^2\ldots (2i-3)^2 (2i-1)$. To do so, we apply this change of variables to the Adler--Moser polynomials to get some polynomials $\tau_n(t_1,t_3,t_5,\ldots)$. Then by seeing them as functions $\tau_n(t_1,t_2,t_3,\ldots)$ of even and odd times we show that they are tau functions of the Kadomstev--Petviasvhili hierarchy (KP). Then since the $\tau_n$'s actually depend only on odd times, they are indeed tau functions of KdV.

It is well known how to compute the polynomial tau functions of KdV without using the Adler--Moser polynomials. For instance, in \cite{Hir04} R.\! Hirota constructs a family of tau functions of KP in terms of Wronskians of the elementary Schur polynomials, which can be reduced to recover the polynomial tau functions of KdV. But the Adler--Moser polynomials reveal a recursive structure in the space of rational solutions of KdV. It would be interesting to investigate how to generalize this to the Drinfeld--Sokolov hierarchies.

	\section{The Adler--Moser Polynomials}

Let a set of evolution variables $\qbf = (x=q_1, q_3, q_5,\ldots)$\footnote{In their original paper \cite{AM78} the authors use the variables $\tau_i$ instead. But since then, the letter $\tau$ has been used rather for the tau functions. Moreover, we use odd indices so that we can later complete the set into the variables of KP.}. We save the variables $\tbf = (x=t_1, t_3, t_5, \ldots)$ for the tau functions of KdV. The Adler--Moser polynomials form a sequence $\theta_n(q_1,q_3\ldots, q_{2n-1})$ for $n\geq 0$, defined by the following recursion \cite{AM78}: $\theta_0 = 1$, $\theta_1 = x$ and for $n\geq 1$,
\begin{equation}\label{eqDiffRecAM}
\theta'_{n+1}\theta_{n-1} - \theta_{n+1}\theta_{n-1}' = (2n-1)\theta_n^2,
\end{equation}
where $f' = \Drm f = \partial_x f$. This recursion leaves an integration constant that is chosen to be $q_{2n-1}$ when computing $\theta_{n}$. We can check that $\theta_n$ has degree $d_n = \frac{1}{2}n(n+1)$ in $x$ and $q_{2n-1}$ is actually the coefficient of $x^{d_{n-2}}$ in this polynomial. The first few polynomials read
\begin{align*}
\theta_0 &= x,\\
\theta_1 &= x,\\
\theta_2 &= x^3 + q_3,\\
\theta_3 &= x^6 + 5q_3x^3 + q_5x - 5q_3^2,\\
\theta_4 &=\textstyle x^{10} + 15q_3x^7 + 7q_5x^5 - 35q_3q_5x^2 + 175q_3^3x - \frac{7}{3}q_5^2 + q_7x^3 + q_3q_7.
\end{align*}
In that same article (\cite{AM78}, pp. 17--18), the authors state the following theorem.

\begin{thm}[Adler--Moser, 1978]\label{thmAM}
There exists a unique change of variables $\qbf\mapsto \tbf$ that transforms the Adler--Moser polynomials $\theta_n(\qbf)$ into the polynomial tau functions $\tau_n(\tbf)$ of KdV. That is, the rational functions $u_n = -2(\log\tau)''$ define operators $L_n = \Drm^2 + u_n$ that satisfy the Lax system of KdV:
\begin{align*}
\frac{\partial L_n}{\partial t_{2i-1}} = \lb \lp L_n^{\frac{2i-1}{2}}\rp_+, L_n \rb.
\end{align*}
\end{thm}

\nid In this article we prove that the desired change of variables is given by: $q_1=t_1=x$, and
\begin{align*}
\sum_{i\geq 2} \frac{q_{2i-1}}{\alpha_{2i-1}} z^{2i-1} = \tanh \lp \sum_{i\geq 2} t_{2i-1} z^{2i-1} \rp.
\end{align*}
where $\alpha_{2i-1} = (-1)^{i-1} 3^25^2\ldots (2i-3)^2 (2i-1)$. The latter coefficients $\alpha_{2i-1}$ where already given in Adler and Moser's article. Notice that this change of variables amounts to simply change the choice of the integration constant in the differential recursion \eqref{eqDiffRecAM}. To prove the statement, let us first recall and prove some results stated in \cite{AM78}. These lemmas relate the Adler--Moser polynomials to a Wronskian representation through a multiplicative factor and a simple rescaling of the variables.\\
\indent Let another set of variables $\mathbf{s} = (x=s_1, s_3,s_5,\ldots)$ and functions $\psi_j(\sbf)$, $j\geq 0$, be defined by:
\begin{align}\label{eqPsiGen}
\sum_{j\geq 1} \psi_j z^{2j-1} = \sinh(xz) + \cosh(xz)\sum_{i\geq 2} s_{2i-1}z^{2i-1},
\end{align}
and $\psi_0 = 0$. It readily implies the recursion $\psi_j'' = \psi_{j-1}$, where we use the notations $\psi_j' = \partial_x\psi_j = \Drm\psi_j$ and $\psi_j^{(i)} = \partial^i_x\psi_j = \Drm^i\psi_j$.

\begin{lem}
The Wronskians of the functions $\psi_j$, defined by
\begin{align*}
W_n := \mathrm{Wr}(\psi_1,\ldots, \psi_n) = \det\lp \Drm^{i-1}\psi_j \rp_{i,j=1,\ldots, n},
\end{align*}
satisfy the recursion
\begin{align}\label{eqDiffRecWronsk}
W_{n+1}'W_{n-1} - W_{n+1}W_{n-1}' = W_n^2.
\end{align}
\end{lem}

\begin{proof}
For any smooth function $\chi(\sbf)$ with respect to $x$, denote $W_n(\chi) := \Wr\lp \psi_1,\ldots, \psi_n, \chi \rp.$ Then one has Jacobi's identity:
\begin{align}\label{eqJacWronsk}
W_n'(\chi)W_{n+1} - W_n(\chi)W_{n+1}' = W_{n+1}(\chi)W_n.
\end{align}
This can be proven by noticing that the left-hand side is a linear differential operator of order $n+1$ which vanishes for $\chi = \psi_1,\ldots, \psi_n$, as well as for $\psi_{n+1}$. Yet the functions $\psi_j$ being linearly independent, the left-hand side must be proportional to $W_{n+1}(\chi)$. Comparing the highest coefficient, one obtain Equation \eqref{eqJacWronsk}. Now thanks to the relation $\psi_j'' = \psi_{j-1}$ and $\psi_1 = x$, we can compute that for $\chi = 1$,
\begin{align}\label{eqWk(1)}
W_n(1) = (-1)^kW_{n-1}.
\end{align}

\nid Then setting $\chi = 1$ in Jacobi's identity \eqref{eqJacWronsk}, one finds Equation \eqref{eqDiffRecWronsk}.
\end{proof}

\nid Now since $W_0 = \theta_0 = 1$ and $W_1 = \theta_1 = x$, the two sequences of polynomials differ only by a multiplicative factor that can be computed:
\begin{align}\label{eqPropAMvsWronsk}
\theta_n(\qbf) = \mu_n W_n(\sbf), & & \mu_n = \prod_{j=1}^k (2k-2j+1)^j.
\end{align}

\begin{lem}
The parameters of the Adler--Moser polynomials $\theta_n (\qbf)$ and those of the Wronskians $W_n(\sbf)$ are related via a rescaling:
\begin{align}\label{eqRescale}
s_{2i-1} = \frac{q_{2i-1}}{\alpha_{2i-1}}, & & \alpha_{2i-1} = (-1)^{i-1} 3^25^2\ldots (2i-3)^2 (2i-1)
\end{align}
\end{lem}

\begin{proof}
It all has to do with the choice of the normalization in the recursions \eqref{eqDiffRecAM} and \eqref{eqDiffRecWronsk}. For $\theta_n$, the choice is such that $q_{2n-1}$ is the coefficient of $x^{d_{n-2}}$. Moreover if $\theta_n$ is a solution of \eqref{eqDiffRecAM}, so is $\theta_n + c\theta_{n-2}$ so that the normalization can be expressed as
\begin{align*}
\theta_n = \mathring{\theta}_n + q_{2n-1}\theta_{n-2},
\end{align*}
where $\mathring{\theta}_n := \left. \theta_n \right|_{q_{2n-1} = 0}.$ Similarly, define $\mathring{W}_n := \left. W_n \right|_{s_{2n-1} = 0}.$ Since $\psi_n = \mathring{\psi_n} + s_{2n-1}$, then by Equation \eqref{eqWk(1)},
\begin{align*}
W_n = \mathring{W_n} + s_{2n-1}\Wr(\psi_1,\ldots,\psi_n,1) = \mathring{W_n} + (-1)^{k-1}s_{2n-1}W_{n-2}.
\end{align*}
Comparing the last two equations and using Equation \eqref{eqPropAMvsWronsk} one obtains the result.
\end{proof}

	\section{A Change of Variables to the Polynomial Tau Functions KdV}

\nid The aim of this section is to prove the following theorem.

\begin{thm}\label{thmChangeVar}
The following change of variables transforms the Adler--Moser polynomials into the polynomial tau functions of KdV: $q_1 = t_1 = x$, and
\begin{align}\label{eqChangeVar}
\sum_{i\geq 2} \frac{q_{2i-1}}{\alpha_{2i-1}} z^{2i-1} = \tanh \lp \sum_{i\geq 2} t_{2i-1} z^{2i-1} \rp.
\end{align}
where $\alpha_{2i-1} = (-1)^{i-1} 3^25^2\ldots (2i-3)^2 (2i-1)$.
\end{thm}

\nid As a matter of fact, this change of variables does not affect the first few variables except for a rescaling. Here are the first variables $q_{2i-1}$ in terms of the $t_{2i-1}$'s: $s_1 = t_1 = x$, and then
\begin{align*}
q_3 &= -3\cdot t_3,	&	q_7 &= -1575\cdot t_7,	&								q_{11} &= -9823275\cdot \left(t_{11} - t_3^2t_5 \right),\\
q_5 &= 45\cdot t_5,	&	q_9 &=\textstyle 99255\cdot \left(t_9 - \frac{1}{3}t_3 \right), &		q_{13} &= 1404728325\cdot \left(t_{13} - t_3^2t_7 - t_3t_5^2 \right).
\end{align*}
\nid Notice that the change of variables is homogeneous for the grading $\deg q_i = \deg t_i = \deg s_i = i$, as proven by Adler and Moser \cite{AM78}. Under this change of variables, we obtain the following polynomial tau functions of KdV:
\begin{align*}
\tau_0 &= 1,\\
\tau_1 &= x,\\
\tau_2 &= x^3 - 3t_3,\\
\tau_3 &= x^6 - 15t_3x^3 - 45t_3^2 + 45t_5x,\\
\tau_4 &= x^{10} -45t_3x^7 + 315t_5x^5 + 4725t_3t_5x^2 - 4725t_3^3x - 4725t_5^2 - 1475t_7x^3 + 4725t_3t_7.
\end{align*}
Then the polynomials $\tau_n/\mu_n$ (cf. Equation \eqref{eqPropAMvsWronsk}) correspond to the tau functions computed via the Wronskians of the elementary Schur polynomials following Hirota \cite{Hir04}.

To prove Theorem \ref{thmChangeVar}, we introduce another sequence of functions defined by the generating series.
\begin{equation}\label{eqGenSeriesPhi}
\sum_{j\geq 1} \phi_j z^{2j-1} = \sinh(xz) + \cosh(xz)\tanh \lp \sum_{i\geq 2} t_{2i-1}z^{2i-1} \rp.
\end{equation}
It amounts to applying the change of variables of Equation \eqref{eqChangeVar} to the functions $\psi_j$ of Equation \eqref{eqPsiGen}. With these functions, define another sequence of Wronskians:
\begin{equation}\label{eqWronskianPhi}
\tau_n:= \Wr(\phi_1,\ldots,\phi_n),
\end{equation}
In what follows, we prove that these Wronskians are tau functions of the KP hierarchy. Yet because they depend only on odd times, they are tau functions of KdV, which proves Theorem \ref{thmChangeVar}. We use the same approach as in \cite{Caf08} and \cite{IZ92}. First we need the following lemma.

\begin{lem}\label{lemRelationPhi}
The functions $\phi_j$ satisfy the following relation for any integers $i, j\geq 1$:
\begin{equation}
\phi_j^{(2i-1)} - \frac{\partial \phi_j}{\partial t_{2i-1}} = \sum_{k = 1}^{j-i-1} \phi_k a_{j-i-k+1},
\end{equation}
where $\phi_j^{(2i-1)} = \partial_x^{2i-1}\phi_j$. Here the $a_j$'s are functions that do not depend on $x$ and are defined by
$$\sum_{j\geq 2} a_jz^{2j-1} = \tanh\lp \sum_{i\geq 2} t_{2i-1}z^{2i-1} \rp.$$
\end{lem}

\begin{proof}
It is a direct calculation: differentiating Equation \eqref{eqGenSeriesPhi} and using a Cauchy product, one obtains
\begin{align*}
\sum_{j\geq 1} \lp \phi_j^{(2i-1)} - \frac{\partial\phi_j}{\partial t_{2i-1}} \rp z^{2j-1} = \sum_{j\geq 1} z^{2j-1}\sum_{k = 0}^{j -i} \phi_k a_{j-i-k+1}.
\end{align*}
Then, noticing that $\phi_0 = 0$ and $a_0 = a_1 = 0$, one gets the correct boundaries in the last sum.
\end{proof}

\begin{rmq}
This relation is to be compared with the one satisfied by the elementary Schur polynomials defined by
\begin{align*}
\exp\lp \sum_{i\geq 1}t_iz^i \rp = \sum_{j\geq 1} P_jz^j.
\end{align*}
These polynomials satisfy the relation $\partial_x^i P_j = \partial_{t_i} P_j.$ Now define $p_j$ to be $P_{2j-1}$ where all even times are set to 0. They satisfy the relation
$$\partial_x^{2i-1} p_j = \partial_{t_{2i-1}} p_j.$$
The last relation allows to prove that their Wronskians $\Wr(p_1,\ldots, p_n)$ are the tau functions of KdV the same way we prove that the Wronskians of the $\phi_j$'s are (see for instance \cite{IZ92}). In particular the Wronskians of the $p_j$'s and those of the $\phi_j$'s coincide.
\end{rmq}

\nid Now let us introduce all the variables $(x=t_1, t_2, t_3, \ldots)$ of the KP hierarchy (that is, odd and even). To prove that the $\tau_n$'s are tau functions of KP we use Sato's equation, which we state in a equivalent form in the following proposition.

\begin{prop}
Let the following differential operator
\begin{align*}
\Delta_n(\chi) := \frac{\Wr\lp \chi, \phi_1, \ldots, \phi_n \rp}{\Wr\lp \phi_1, \ldots, \phi_n \rp} = \frac{1}{\tau_n}\Wr\lp \chi, \phi_1, \ldots, \phi_n \rp,
\end{align*}
for any differentiable function $\chi$ with respect to $x$. The following equation holds for any $i\geq 1$:
\begin{equation}\label{eqPreSato}
\frac{\partial \Delta_n}{\partial t_{i}} = \lp \Delta_n \Drm^{i} \Delta_n^{-1} \rp_+ \Delta_n - \Delta_n \Drm^{i}.
\end{equation}
\end{prop}

\begin{proof}
It is sufficient to prove the equality when acting on $\phi_1,\ldots, \phi_n$ which are $n$ linearly independent functions. Yet these functions are solutions of the equation $\Delta_n (\phi_j) = 0,$ so it amounts to proving that
\begin{align*}
\frac{\partial \Delta_n}{\partial t_i} \lp\phi_j\rp + \Delta_n \lp \phi_j^{(i)} \rp = 0.
\end{align*}
If $i = 2\ell$ is even, then we only have to prove that $\Delta_n \lp \phi_j^{(2\ell)} \rp = 0.$ Yet $\phi_j'' = \phi_j$ so that $\phi_j^{(2\ell)} = \phi_{j-\ell}$, or 0 if $j\leq \ell$. Eventually, it amounts to $\Delta_n\lp\phi_{j-\ell} \rp = 0$ which holds true. If $i = 2\ell-1$ is odd, by Lemma \ref{lemRelationPhi}, it amounts to proving that
\begin{align*}
\frac{\partial \Delta_n}{\partial t_{2\ell-1}} \lp\phi_j\rp + \Delta_n \frac{\partial}{\partial t_{2\ell-1}}\lp\phi_j\rp  + \sum_{k = 1}^{j-\ell-1} \Delta_n \lp \phi_k a_{j-\ell-k+1} \rp = 0.
\end{align*}
Yet the functions $a_j$ do not depend on $x$, so the last sum vanishes. And for the other terms it amounts to
\begin{align*}
\frac{\partial \Delta_n}{\partial t_{2\ell-1}} \lp\phi_j\rp+ \Delta_n \frac{\partial}{\partial t_{2\ell-1}}\lp\phi_j\rp = \frac{\partial}{\partial t_{2\ell-1}}\Delta_n \lp\phi_j \rp = 0.
\end{align*}
\end{proof}

\nid Now the fact that Equation \eqref{eqPreSato} implies the KP hierarchy via Sato's Equation is a classical result which we sum up in the following proposition.

\begin{prop}
Let the pseudo-differential operator $S_n = \Delta_n\Drm^{-n}$. Then $S_n$ satisfies Sato's equation
\begin{align}\label{eqSato}
\frac{\partial S_n}{\partial t_i} + \lp S_n\Drm^i S_n^{-1} \rp_- S_n = 0.
\end{align}
Moreover, the pseudo-differential operator $\mathcal{L}_n = S_n\Drm S_n^{-1}$ is of the form $\mathcal{L}_n = \Drm + \sum_{k\geq 1} f_{n,k}\Drm^{-k}$ and satisfies the Lax system of the KP hierarchy
\begin{equation}\label{eqKPLax}
\frac{\partial \mathcal{L}_n}{\partial t_i} = \left[ \lp \mathcal{L}_n^i \rp_+, \mathcal{L}_n \right].
\end{equation}
\end{prop}

\begin{proof}
Equation \eqref{eqPreSato} on differential operators readily implies Sato's equation \eqref{eqSato} on purely pseudo-differential operators. Then it is well known how to derive the Lax system of KP from Sato's equation. One only has to differentiate the relation $\mathcal{L}_nS_n = S_n\Drm$ by $t_i$, then use the relation $\lb \mathcal{L}_n, \mathcal{L}_n^{i}\rb = 0$.
\end{proof}

\nid Finally, the following proposition states that the dressing operator $S_n$ is indeed related to the Wronskians $\tau_n$ via the usual shift equation, ie, that $\tau_n$ is a tau function of KP.

\begin{prop}\label{propShift}
The dressing operator $S_n$ and the Wronskian $\tau_n$ satisfy the following relation
\begin{equation}\label{eqWaveFunction}
S_n\lp \e^{\xi(\tbf;\lambda)} \rp = \e^{\xi(\tbf;\lambda)}\frac{\tau_n\lp \tbf - [\lambda^{-1}] \rp}{\tau_n(\tbf)},
\end{equation}
where $\xi(\tbf;\lambda) = \sum_{i\geq 1} t_i\lambda^i$ and 
\begin{equation}\label{eqShift}
\tau_n \lp \tbf - [\lambda^{-1}] \rp := \textstyle \tau_n \lp t_1 - \frac{1}{\lambda}, t_3 - \frac{1}{3\lambda^3}, t_5 - \frac{1}{5\lambda^5}, \ldots \rp.
\end{equation}
\end{prop}

\nid To prove that, we need the following lemma.

\begin{lem}\label{lemTriangRel}
The shift of the functions $\phi_j$ reads the following triangular relation for any $j\geq 1$,
\begin{equation}\label{eqTriangRelPhi}
\phi_j(\tbf - [\lambda^{-1}]) = \phi_j(\tbf) - \lambda^{-1}\phi_j'(\tbf) + \sum_{i = 1}^{j-1} \lp \phi_i(\tbf) - \lambda^{-1}\phi_i'(\tbf) \rp b_{j-i}(\tbf),
\end{equation}
where the $b_j$'s are functions that do not depend on $x$ and are defined by
\begin{equation}\label{eqGenSerB}
\sum_{j\geq 0} b_jz^{2j} = \mathrm{sech} \lp z\lambda^{-1} \rp \lb 1 - z\lambda^{-1}\tanh \lp z\lambda^{-1}\rp - z\lambda^{-1}\tanh \lp \eta \rp + \tanh\lp z\lambda^{-1}\rp \tanh \lp \eta \rp \rb^{-1},
\end{equation}
with $b_0 = 1$, and
$$\eta(\tbf;z) := \sum_{i\geq 2} t_{2i-1}z^{2i-1}.$$
Here $\mathrm{sech} = 1/\cosh$ and the exponent $-1$ stands for the multiplicative inverse of formal power series.
\end{lem}

\begin{proof}
Using the fact that 
$$\tanh^{-1}(z\lambda^{-1}) = z\lambda^{-1} + \sum_{i\geq 2} \frac{z^{2i-1}}{(2i-1)\lambda^{2i-1}}$$
for $z\lambda^{-1}\in (-1,1)$, and applying the sum formulae of hyperbolic functions, one obtains that
\begin{align*}
\sum_{j\geq 1} \phi_j\lp \tbf - [\lambda^{-1}] \rp z^{2j-1} = \frac{\mathrm{sech}\lp z\lambda^{-1}\rp \sum_{j\geq 1} \lp \phi_j - \lambda^{-1}\phi_j' \rp z^{2j-1} }{1 - z\lambda^{-1}\tanh \lp z\lambda^{-1}\rp - z\lambda^{-1}\tanh \lp \eta \rp + \tanh\lp z\lambda^{-1}\rp \tanh \lp \eta \rp}.
\end{align*}
Moreover, the above denominator has its constant term equal to 1 and only even powers, so is its multiplicative inverse. Therefore, the series $\sum_{j\geq 0} b_j z^{2j}$ in Equation \eqref{eqGenSerB} is well defined and has indeed $b_0 = 1$, hence the triangular relation \eqref{eqTriangRelPhi}.
\end{proof}

\begin{rmq}
As in Lemma \eqref{lemRelationPhi}, this relation is to be compared with the one satisfied by the elementary Schur polynomials, namely,
\begin{align*}
P_j(\tbf - [\lambda^{-1}]) = P_j(\tbf) - \lambda^{-1}P_{j-1}(\tbf).
\end{align*}
\end{rmq}

\nid We can now prove Proposition \ref{propShift} which concludes the proof of Theorem \ref{thmChangeVar}.

\begin{proof}[Proof of Proposition \ref{propShift}]
We prove an equivalent equation which only uses differential operator:
\begin{align*}
\Delta_n \lp \e^{\xi(\tbf;\lambda)} \rp = \lambda^n\e^{\xi(\tbf;\lambda)} \frac{\tau_n\lp \tbf - [\lambda^{-1}] \rp}{\tau_n(\tbf)}.
\end{align*}
Thanks to Lemma \ref{lemTriangRel}, we can rewrite the right-hand side as
\begin{align*}
\frac{\lambda^n\e^{\xi(\tbf;\lambda)}}{\tau_n} \left| \begin{matrix}
\phi_1 - \lambda^{-1}\phi_1' & \cdots & \phi_1^{(n-1)} - \lambda^{-1}\phi_1^{(n)}\\
\vdots & & \vdots\\
\phi_n - \lambda^{-1}\phi_n' + \sum_{i=1}^{j-1}\lp \phi_i - \lambda^{-1}\phi_i'\rp b_{j-i} & \cdots & \phi_n^{(n-1)} - \lambda^{-1}\phi_n^{(n)} + \sum_{i=1}^{j-1}\lp \phi_i^{(n-1)} - \lambda^{-1}\phi_i^{(n)}\rp b_{j-i}
\end{matrix} \right|.
\end{align*}
On the other hand, the left-hand side reads
\begin{align*}
\Delta_n\lp \e^{\xi(\tbf;\lambda)} \rp = \frac{1}{\tau_n} \left| \begin{matrix}
\e^{\xi(\tbf;\lambda)} & \lambda\e^{\xi(\tbf;\lambda)} & \cdots & \lambda^n \e^{\xi(\tbf;\lambda)}\\
\phi_1 & \phi_1' & \cdots & \phi_1^{(n)}\\
\vdots & & & \vdots\\
\phi_n & \phi_n' & \cdots & \phi_n^{(n)}
\end{matrix} \right|.
\end{align*}
Using operations on rows and columns, we can easily check that these two expressions are equal.
\end{proof}~\\

\nid \textbf{Acknowledgements.} The author would like to warmly thank M. Cafasso as well as M. Adler for their help and guidance. Besides, the author would like to express their most sincere thanks to J. Leray for his most useful insight in recognizing the Taylor series expansion of hyperbolic tangent.

\bibliography{mybib}
\bibliographystyle{alpha}

\end{document}